%% file: corner1.tex
\documentclass[11pt]{amsart}

\usepackage{amsmath}
\usepackage{amsthm}
\usepackage{mathtools}
\usepackage{amssymb}
\usepackage{enumerate}
\usepackage{slashed}
\usepackage{graphicx}
\usepackage{newlfont}
\usepackage{amsrefs}
\usepackage{comment}
\usepackage[normalem]{ulem}
\usepackage{mathtools}
\usepackage{accents}
\usepackage{leftidx}
\usepackage[marginpar=18mm,top=28mm,bottom=28mm,right=20mm,left=20mm]{geometry}
\usepackage{xcolor}
\usepackage[normalem]{ulem}

\numberwithin{equation}{section}

\theoremstyle{plain}
\newtheorem{theorem}{Theorem}

\newtheorem{cor}[theorem]{Corollary}

\theoremstyle{definition}

\newtheorem{XxmpX}{Example} 
\newenvironment{example}    
{
	\pushQED{\qed}\begin{XxmpX}}
	{\popQED\end{XxmpX}}

\theoremstyle{remark}
\newtheorem{remark}{Remark}

\newcommand{\tr}{\operatorname{tr}}

\newcommand{\R}{\mathbb{R}}

\newcommand{\Z}{\mathbb{Z}}

\begin{document}

\title[Static Vacuum Space-Periodic Solutions]
{5-Dimensional Space-Periodic Solutions of the Static Vacuum Einstein Equations}

\author[Khuri]{Marcus Khuri}
\address{Department of Mathematics\\
Stony Brook University\\
Stony Brook, NY 11794, USA}
\email{khuri@math.sunysb.edu}

\author[Weinstein]{Gilbert Weinstein}
\address{Physics Department and Department of Mathematics\\
Ariel University\\
Ariel, 40700, Israel}
\email{gilbertw@ariel.ac.il}

\author[Yamada]{Sumio Yamada}
\address{Department of Mathematics\\
Gakushuin University\\
Tokyo 171-8588, Japan}
\email{yamada@math.gakushuin.ac.jp}

\thanks{M. Khuri acknowledges the support of NSF Grant DMS-1708798 and Simons Foundation Fellowship 681443. S. Yamada acknowledges the support of JSPS Grant KAKENHI 17H01091.}

\begin{abstract}
An affirmative answer is given to a conjecture of Myers concerning the existence of 5-dimensional regular static vacuum solutions that balance an infinite number of black holes, which have Kasner asymptotics. A variety of examples are constructed, having different combinations of ring $S^1\times S^2$ and sphere $S^3$ cross-sectional horizon topologies.
Furthermore, we show the existence of 5-dimensional vacuum \textit{solitons} with Kasner asymptotics. These are regular static space-periodic vacuum spacetimes devoid of black holes. Consequently, we also obtain new examples of complete Riemannian manifolds of nonnegative Ricci curvature in dimension 4, and zero Ricci curvature in dimension 5, having arbitrarily large as well as infinite second Betti number.
\end{abstract}
\maketitle

\input{intro1}

\input{background1}

\input{stateresults1}

\input{kn1}

\input{reg1}

\bibliographystyle{amsplain}
\bibliography{mybib}

\end{document}

%% file: intro1.tex
\section{Introduction}

The Majumdar-Papapetrou solutions \cite{Majumdar,Papapetrou} of the 4D Einstein-Maxwell equations show that gravitational attraction and electro-magnetic repulsion may be balanced in order to support multiple black holes in static equilibrium. In \cite{myers1987}, Myers generalized these solutions to higher dimensions, and also showed that such balancing is also possible for vacuum black holes in 4 dimensions if an infinite number are aligned properly in a periodic fashion along an axis of symmetry; these 4-dimensional solutions were later rediscovered by Korotkin and Nicolai in \cite{korotkinnicolai}. Although the Myers-Korotkin-Nicolai solutions are not asymptotically flat, but rather asymptotically Kasner, they play an integral part in an extended version of static black hole uniqueness given by Peraza and Reiris \cite{PerazaReiris}. It was conjectured in \cite{myers1987} that these vacuum solutions could be generalized to higher dimensions, possibly with black holes of nontrivial topology. One of the aims of the current paper is to confirm this to be the case. Namely, a variety of examples are produced of 5D space-periodic static vacuum solutions having combinations of the ring $S^1\times S^2$ and sphere $S^3$ horizon cross-sectional topologies.

Typically multiple black holes in static vacuum require conical singularities in order to remain in equilibrium, with the singularity being interpreted as a force holding the black holes apart.  A detailed analysis of this situation was originally carried out by Bach and Weyl \cite{bachweyl} for the 4-dimensional case. In 5 dimensions the situation is complicated by the fact that multiple axes may be separated by corner points, where two linearly independent Killing fields vanish. Moreover we note that, unlike in 4D, the force along an axis is not necessarily constant in higher dimensions. 
Nevertheless, balancing is shown to hold for our examples due to periodicity and symmetry.

An asymptotically flat stationary vacuum solution which is geodesically complete, and in particular devoid of horizons, must be Minkowski space. In 4 dimensions this is a classical result of Lichnerowicz \cite{Lichnerowicz}. A modern treatment, which holds in all dimensions, shows that this follows from the rigidity statement of the positive mass theorem \cite{schoenyau1979,SchoenYau2017,witten1981} combined with the Komar expression for mass and Stokes' theorem.
Similar results hold for the Einstein-Maxwell equations when $D=4$. However, in higher dimensional supergravity theories, horizonless soliton solutions can possess nontrivial 2-cycle `bubbles' that are supported by magnetic flux supplied by Maxwell fields \cite{Kundurietal}. Here, a \textit{gravitational soliton} refers to a nontrivial, globally stationary, geodesically complete spacetime.
In the 4-dimensional vacuum case, it turns out that
the no-solitons result above essentially holds even without the assumption of asymptotic flatness \cite[Theorem 0.1]{Anderson}, the only difference being that the resulting spacetime may be a nontrivial quotient of Minkowski space instead of Minkowski space itself. It may then seem surprising that vacuum solitons do exist in 5 dimensions. In fact, we will show that there are geodesically complete non-flat space-periodic solutions to the bi-axisymmetric static vacuum equations. These examples may be constructed so that the time slices have topology given by the connected sums
\begin{equation}
\#^k \> S^2\times S^2,
\end{equation}
for any $k\geq 1$ including $k=\infty$. Furthermore, as a corollary of these results we are able to produce new examples of complete Riemannian manifolds of nonnegative Ricci curvature in dimension 4, and zero Ricci curvature in dimension 5, having arbitrarily large as well as infinite second Betti number.

This paper is organized as follows. In the next section background material needed to state the main results is given, while Section \ref{sec3} is dedicated to the statement of the main results and further discussion. In Section \ref{kn}, higher dimensional generalizations of the Myers-Korotkin-Nicolai solutions are given, answering a conjecture of Myers. The construction of space-periodic soliton solutions is also detailed in this section.
Lastly, analysis of the regularity and asymptotic properties of the solutions is given in Section \ref{regularity}.

%% file: background1.tex
\section{Background}

Let $\mathcal{M}^5$ be the domain of outer communication of a stationary bi-axisymmetric 5-dimensional spacetime. The orbit space $\mathcal{M}^5/[\mathbb{R}\times U(1)^2]$, with mild hypotheses \cite{hollands2011}, is homeomorphic to the right half plane $\{(\rho,z)\mid \rho>0\}$. The vacuum Einstein equations in this setting reduce to an axisymmetric harmonic map, with domain $\mathbb{R}^3\setminus \{z-\text{axis}\}$ parameterized by the cylindrical coordinates $(\rho,z,\varphi)$. The $z$-axis is decomposed into an exhaustive sequence of intervals referred to as \textit{rods}, and denoted by $\{\Gamma_l\}_{l\in I}$ for some index set $I$ which may be infinite. Rods are divided into two types, \textit{axis rods} and \textit{horizons rods}. Each \textit{axis rod} $\Gamma_l$ is defined by the vanishing of a linear combination $p_l \partial_{\phi^1}+q_l \partial_{\phi^2}$ of the generators $\partial_{\phi^1}$, $\partial_{\phi^2}$ of the $U(1)^2$ symmetry, where the tuple $(p_l,q_l)$ of relatively prime integers is called the \textit{rod structure} of $\Gamma_l$. By contrast a horizon rod $\Gamma_h$ is an interval of the $z$-axis where no closed-orbit Killing field degenerates, that is $p_h=q_h=0$, but where $|\partial_t+\Omega_1\partial_{\phi^1}+\Omega_2\partial_{\phi^2}|$ vanishes with $\partial_t$ representing the stationary Killing field, and $\Omega_1$, $\Omega_2$ denoting angular velocities of the horizon. The intersection point of two neighboring axis rods is called a \textit{corner} and has the property that both rotational Killing fields vanish, while the intersection point of an axis rod with a horizon rod is called a \textit{pole}. In order to
avoid orbifold singularities at corner points, neighboring axis rod structures must satisfy the \textit{admissibility condition}
\begin{equation}
\det\begin{pmatrix} p_l & q_l \\ p_{l+1} & q_{l+1} \end{pmatrix} = \pm 1.
\end{equation}
The collection of rods and associated rod structures completely determines the topology of horizons and the domain of outer communication, see \cite{khurimatsumotoweinsteinyamada} for a detailed analysis. In particular, if the two axis rods surrounding a horizon rod have rod structures $(1,0)$, $(0,1)$ then the horizon topology is that of a sphere $S^3$, whereas $(1,0)$, $(1,0)$ yields a ring $S^1\times S^2$, and $(1,0)$, $(q,p)$ produces the lens space $L(p,q)$. These are the only possible topologies for horizon cross-sections in the 5D stationary bi-axisymmetric setting.

Recall that the stationary bi-axisymmetric vacuum Einstein equations \cite{Hollands2012,shiromizu} reduce to solving the following harmonic map equations
\begin{align}\label{equationshm}
\begin{split}
&\Delta f_{ij}-f^{nm}\nabla^{k}f_{im}\nabla_{k}f_{jn}
+f^{-1}\nabla^{n}\omega_{i}\nabla_{n}\omega_{j}=0,\\
&\Delta\omega_{i}-f^{nm}\nabla^{k}f_{in}\nabla_{k}\omega_{m}
-f^{nm}\nabla^{k}f_{nm}\nabla_{k}\omega_{i}=0,
\end{split}
\end{align}
where $F=(f_{ij})$ is a $2\times 2$ symmetric matrix which is positive definite away from the axes, $f=\det F$, and $\omega=(\omega_1,\omega_2)^{t}$ are twist potentials associated with the $U(1)^2$ symmetry. The spacetime metric on $\mathcal{M}^5$ can be constructed from these quantities and expressed in Weyl-Papapetrou coordinates by
\begin{equation}\label{spacetimem}
g=e^{2\alpha}(d\rho^2+dz^2)-f^{-1}\rho^2 dt^2
+f_{ij}(d\phi^{i}+v^{i}dt)(d\phi^{j}+v^{j}dt).
\end{equation}
Note that this exhibits the interpretation of rod structures as vectors $(p_l,q_l)^t$ lying in the (1-dimensional) kernel of the matrix $F$ at an axis rod $\Gamma_l$.
The functions $v^i$ and $\alpha$ may be obtained by quadrature, in particular
they are obtained by integrating the equations
\begin{equation}\label{vequation}
v^i_{\rho}=\rho f^{-1}f^{ij}\omega_{j,z},\quad\quad\quad
v^i_{z}=-\rho f^{-1}f^{ij}\omega_{j,\rho},
\end{equation}
and
\begin{align}\label{alpha}
\begin{split}
	\alpha_\rho =&
	\frac\rho8 \left[ (\log f)_\rho^2 - (\log f)_z^2 + \tr F^{-1}F_\rho F^{-1}F_\rho
	- \tr F^{-1}F_z F^{-1}F_z
 - \frac4\rho (\log f)_\rho + \frac2{f}F^{-1}(\omega_\rho^2-\omega_z^2) \right], \\
	\alpha_z =& \frac\rho4 \left[ (\log f)_\rho(\log f)_z + \tr F^{-1}F_\rho F^{-1}F_z	 -\frac2\rho (\log f)_z + \frac2{f} F^{-1}\omega_{\rho}\omega_z \right],
\end{split}
\end{align}
where we have used the notation $F^{-1}\mathbf{v}\mathbf{w}:=\mathbf{v}^t F^{-1}\mathbf{w}$ for vectors $\mathbf{v}$ and $\mathbf{w}$. The integrability conditions for \eqref{vequation} and \eqref{alpha} correspond to the harmonic map equations \eqref{equationshm}. If $\Gamma$ denotes the union of all axis rods, then the relevant harmonic map $\Phi:\mathbb{R}^3\setminus\Gamma\rightarrow SL(3,\mathbb{R})/SO(3)$ is built from $(F,\omega)$ and is represented by a $3\times 3$ symmetric positive definite unimodular matrix \cite{khuriweinsteinyamada2017,maison1979ehlers}. Boundary conditions are imposed on the axis in order to achieve the desired rod structures, and the potentials $\omega$ are prescribed to be constants $\mathbf{c}_l\in\R^2$ on each axis rod $\Gamma_l$, such that the values of the constants agree on consecutive axis rods. Hence, the potential constants only change across horizon rods, and the difference determines the horizon angular momentum.

Due to the axisymmetries, it is possible that conical singularities form on the axes when constructing the spacetime metric \eqref{spacetimem}. The conical singularity at a point $(0,z_0)$ on an axis rod $\Gamma_{i}$, with rod structure $\mathbf{v}=(v^1,v^2)=(p_l,q_l)$, is measured by the angle defect $\theta\in(-\infty,2\pi)$ associated with the 2-dimensional cone formed by the orbits of $v^j\partial_{\phi^j}$ over the line $z=z_0$. This may be expressed as
\begin{equation}\label{conicalformula}
\frac{2\pi}{2\pi-\theta}=\lim_{\rho\rightarrow 0}\frac{2\pi\cdot\mathrm{Radius}}{\mathrm{Circumference}}
=\lim_{\rho\rightarrow 0}
\frac{\int_{0}^{\rho}\sqrt{e^{2\alpha}}}{\sqrt{f_{ij}v^{i}v^{j}}}
=\lim_{\rho\rightarrow 0}
\sqrt{\frac{\rho^2 e^{2\alpha}}{f_{ij}v^{i}v^{j}}}.
\end{equation}
A zero angle defect indicates the absence of a conical singularity.
In this case, with the aid of a change of coordinates from polar to Cartesian, it is straightforward to check that this condition is necessary and sufficient for the metric to be extendable across the axis, assuming that analytic regularity has been established. We denote by $b_l$ the value of $\log\left(\frac{2\pi}{2\pi-\theta}\right)$ on the axis rod $\Gamma_l$, and will refer to this quantity as the \textit{logarithmic angle defect}. This is well defined since the angle defect is constant on each axis rod \cite{chenteo,harmark2004stationary}. The conical singularity on $\Gamma_l$ is said to exhibit an \textit{angle deficit} if $b_l>0$, and an \textit{angle surplus} if $b_l<0$.

%% file: stateresults1.tex
\section{Statement of Results}\label{sec3}

The Myers-Korotkin-Nicolai solutions \cite{myers1987,korotkinnicolai} are axisymmetric static 4-dimensional spacetimes balancing an infinite number of nondegenerate black holes, which are strung along the axis of rotation at periodic intervals. The domain of outer communication (DOC) has topology $\mathcal{M}^4=\mathbb{R}\times M^3$ with
$M^3=\mathbb{R}^3\setminus \cup_{i=1}^{\infty}B_i^3$, where each 3-ball $B_i^3$ represents a single black hole. These solutions are space-periodic, namely the group $\mathbb{Z}$ acts by isometries through translations in the $z$-direction of the Weyl-Papapetrou coordinate system. By taking quotients with different subgroups, solutions are obtained on DOCs having slice topology $M^3=S^3\setminus\left(S^1\times B^2 \cup_{i=1}^{i_0} B_i^3\right)$ where $i_0\geq 1$ is finite and $B^2$ denotes the 2-dimensional disk. We note that this is equivalent to the complement of $i_0$ disjoint 3-balls in a solid 3-torus.

In \cite{myers1987} Myers conjectured that analogues of these static vacuum solutions are possible in higher dimensions, perhaps with black holes of nontrivial topology. Here we show that indeed this is the case, by exhibiting a variety of such solutions. The notion of an infinite connected sum plays a role. Without further clarification this concept can be ambiguous, and so in Section \ref{kn} a precise definition is provided.

\begin{theorem}\label{kng}
There exist regular bi-axisymmetric solutions of the 5D static vacuum Einstein equations, balancing an infinite number of spherical $S^3$ and ring $S^1\times S^2$ black holes in various combinations. These solutions are space-periodic, asymptotically Kasner, and have DOC topology $\mathcal{M}^5=\mathbb{R}\times M^4$ where
\begin{equation}\label{ty}
M^4=
\#^{\infty}\> S^2\times S^2 \setminus\left[\cup_{i=1}^{i_0} B_i^4 \cup_{j=1}^{j_0} (S^2\times B^2)_j\right],
\end{equation}
and $i_0,j_0\in\{0,\infty\}$, with $0$ indicating the absence of the summands.
Furthermore, taking quotients with appropriate subgroups of $\mathbb{Z}$ gives rise to solutions on DOCs having slice topology
\begin{equation}\label{ty1}
M^4/\!\sim \!\!\!\!\text{ }=M_{k}^4 \setminus\left[(T^2\times B^2) \cup_{i=1}^{i_0}B_i^4 \cup_{j=1}^{j_0}(S^2\times B^2)_j\right],
\end{equation}
where $0\leq i_0,j_0<\infty$ and $M_{k}^4$ is either $S^4$ or $\operatorname \#^{k} S^2\times S^2$ depending on whether $k=0$ or $1\leq k<\infty$. The summands of \eqref{ty} and \eqref{ty1} are all mutually disjoint.
\end{theorem}

Each of the $B^4_i$ and $(S^2\times B^2)_j$ in \eqref{ty}, \eqref{ty1} represent $S^3$ and
$S^1\times S^2$ black hole cross-sections, respectively, whereas $T^2\times B^2$ represents an asymptotic end in a time slice having topology $\mathbb{R}\times T^3$. The statement of this theorem, and in particular (\ref{ty}) and (\ref{ty1}), is based on examples explicitly constructed in Section \ref{kn}. However, the exact fashion in which the excisions of $B^4$, $S^2 \times B^2$, and $T^2 \times B^2$  are performed, is intentionally left ambiguous for now.

The spacetime metrics of Theorem \ref{kng} have the following asymptotic structure in a modified Weyl-Papapetrou coordinate system
\begin{equation}
g\thicksim d\tau^2+\tau^{2p_1}dz^2 -\tau^{2p_2}dt^2
+\tau^{2p_3}(d\phi^1)^2 +\tau^{2p_4}(d\phi^2)^2,
\end{equation}
where the exponents satisfy the \textit{Kasner conditions} $\sum_{\ell} p_{\ell}=\sum_{\ell} p_{\ell}^2=1$.
For this reason, in analogy with the Myers-Korotkin-Nicolai metrics, the solutions of Theorem \ref{kng} are referred to as asymptotically Kasner. In addition, each black hole is balanced by ensuring that the total force exerted by all the bodies to either side along the axes sums to zero. Mathematically, we achieve this by showing that conical singularities do not form on the axes, due to the symmetry and periodicity of the resulting spacetime metrics. Interestingly, these constructions can be used to show that regular horizonless static vacuum (nontrivial) solutions exist in 5 dimensions. As discussed in the introduction, this is surprising since similar results do not hold in 4 dimensions and all previously known 5D soliton solutions require certain types of matter fields to support their equilibrium state.

\begin{theorem}\label{nobody}
There exist regular complete bi-axisymmetric solutions of the 5D static vacuum Einstein equations, which are devoid of black holes. These solutions are space-periodic, asymptotically Kasner, and have DOC topology $\mathcal{M}^5=\mathbb{R}\times M^4$ where
\begin{equation}\label{infinite}
M^4=\#^{\infty}\> S^2\times S^2.
\end{equation}
Furthermore, given a nonnegative integer $k$, a quotient may be taken by $k\mathbb{Z}$ to obtain solutions on DOCs having slice topology
\begin{equation} \label{time-slice}
M^4/\!\sim \!\!\!\!\text{ }=M_{k}^4 \setminus(T^2\times B^2),
\end{equation}
where $M_{k}^4$ is either $S^4$ or $\operatorname \#^{k} S^2\times S^2$ depending on whether $k =0$ or $1\leq k <\infty$.
\end{theorem}

These soliton solutions may be parlayed into new examples of complete Riemannian manifolds of nonnegative Ricci curvature in dimension 4, and zero
Ricci curvature in dimension 5, having arbitrarily large, as well as infinite second Betti number. Previous examples with positive Ricci curvature have been constructed by Sha and Yang \cite{ShaYang} on related topologies, via different methods.
In our setting, consider a static spacetime $(\mathbb{R}\times M^4, -w^2 dt^2 +\bar{g})$ satisfying the vacuum Einstein equations
\begin{equation}
\bar{R}_{ij}=\frac{\nabla_{ij} w}{w},\quad\quad\quad\quad \Delta_{\bar{g}} w=0.
\end{equation}
The conformal metric $\tilde{g}=w\bar{g}$ then has nonnegative Ricci curvature. To see this, recall how the Ricci tensor changes under a conformal deformation. Namely, if $\tilde{g}=e^{2\varphi}\bar{g}$ then
\begin{equation}
\tilde{R}_{ij}=\bar{R}_{ij}-2\left(\nabla_{ij}\varphi -\varphi_i \varphi_j\right)
-\left(\Delta_{\bar{g}}\varphi+2|\nabla \varphi|^2\right)\bar{g}_{ij}.
\end{equation}
Setting $\varphi=\frac{1}{2}\log w$ produces
\begin{equation}
\tilde{R}_{ij}=\frac{3}{2}(\log w)_i (\log w)_j.
\end{equation}
This may then be applied to the solitons of the previous section, where the function $w$ is smooth, space-periodic, and strictly positive. The conformally deformed time slice metrics of the solitons then provide examples of nonnegative Ricci curvature on the topologies $\operatorname \#^{\infty} S^2\times S^2$, and $\operatorname \#^{k} S^2\times S^2 \setminus (T^2\times B^2)$ for any $k\geq 0$.
Lastly, consider the soliton spacetimes and perform a Wick rotation $t\mapsto \sqrt{-1}s$, where identifications are then instituted along the $s$-coordinate so that it may be viewed as parameterizing a circle. This transformation preserves the Ricci flat condition. We then obtain complete Ricci flat metrics on $S^1\times \operatorname \#^{\infty} S^2\times S^2$, and $S^1\times \left[\operatorname \#^{k} S^2\times S^2 \setminus (T^2\times B^2)\right]$ for any $k\geq 0$. The following corollary summarizes this discussion.

\begin{cor}\label{corollary11}
$(i)$ There exists a smooth complete bi-axisymmetric Riemannian metric of nonnegative Ricci curvature on the following manifold of infinite topological type $\operatorname \#^{\infty} S^2\times S^2$.
Furthermore the metric is periodic, so that taking quotients as in Theorem \ref{nobody} produces metrics of the same type on
\begin{equation}
\operatorname \#^{k} S^2\times S^2 \setminus (T^2\times B^2),
\end{equation}
for any $k\geq 0$.

$(ii)$ There exists a smooth complete tri-axisymmetric Riemannian metric of zero Ricci curvature on the following manifold of infinite topological type
$S^1\times \operatorname\#^{\infty} S^2\times S^2$.
Furthermore the metric is periodic, so that taking quotients as in Theorem \ref{nobody} produces metrics of the same type on
\begin{equation}
S^1\times \left[\operatorname \#^{k} S^2\times S^2 \setminus (T^2\times B^2)\right],
\end{equation}
for any $k\geq 0$.
\end{cor}

%% file: kn1.tex
\section{Space-Periodic Solutions}
\label{kn}

In this section we will present a series of examples illustrating Theorems \ref{kng} and \ref{nobody}.
All of the examples presented below are 5-dimensional bi-axisymmetric static vacuum solutions expressed in Weyl-Papapetrou coordinates \eqref{spacetimem}.
The metric along the torus fibers will be given by a diagonal matrix function
$F = (f_{ij})$, namely
\begin{equation}\label{F11}
	F=\begin{pmatrix}
	e^u & 0 \\
	0 & e^v
	\end{pmatrix},
\end{equation}
where $u$ and $v$ are harmonic functions on $\mathbb{R}^3\setminus \Gamma$ so that the harmonic map equations \eqref{equationshm} are satisfied (note that $|\omega|=0$ in the static setting). This ansatz implies that the axes can only exhibit the two rod structures $(1,0)$ or $(0,1)$. For an axis rod $\Gamma_1$ having the rod structure $(1,0)$, and an axis rod $\Gamma_2$ having the rod structure $(0,1)$ we find that the corresponding logarithmic angle defects \eqref{conicalformula} are given by
\begin{equation}\label{b1-b2}
b_1 =  \lim_{\rho \rightarrow 0}  \left(\log \rho +\alpha-\frac{1}{2}u\right)\quad\text{ on }\quad\Gamma_1, \quad \quad\quad  b_2 = \lim_{\rho \rightarrow 0} \left(\log \rho+\alpha-\frac{1}{2}v\right)\quad\text{ on }\quad\Gamma_2 .
\end{equation}
In the examples below, the functions $u$ and $v$ will be periodic in the $z$-direction by construction. The function $\alpha$, obtained by quadrature from \eqref{alpha}, will then be shown to also possess the same periodicity, yielding the desire space-periodic static vacuum spacetimes. Geometric regularity of the solutions is established by eliminating the possibility of conical singularities along the axes. This is achieved by utilizing the three degrees of freedom obtained by adding constants to $u$, $v$, and $\alpha$.

Let $z_a=z-a$ and $r_a=\sqrt{\rho^2+z_a^2}$. In what follows, the potential
\begin{equation}
U_I = \log(r_a-z_a) - \log(r_b-z_b)
\end{equation}
for a uniform charge distribution along a finite interval $I=[a,b]$ of the $z$-axis, will be utilized repeatedly. Observe that this function satisfies the properties
\begin{equation}\label{kilo}
U_I<0,\quad\quad\quad U_I \sim 2\log\rho\quad\text{ near }\quad I,\quad\quad\quad
U_I=(a-b)/r + O(r^2)\quad\text{ as }\quad r\to\infty.
\end{equation}
Furthermore, due to the diagonal structure of $F$, the equations \eqref{alpha} defining $\alpha$ simplify to
\begin{align}\label{alpha1}
\begin{split}
\alpha_\rho =&
\frac\rho4 \left[ u_\rho^2 - u_z^2 + v_\rho^2 - v_z^2 + u_\rho v_\rho - u_z v_z - \frac2\rho (u_\rho+v_\rho) \right], \\
\alpha_z =& \frac\rho4 \left[ 2u_\rho u_z + 2v_\rho v_z + u_\rho v_z + u_z v_\rho - \frac2\rho(u_z+v_z) \right].
\end{split}
\end{align}

\begin{example}[String of black rings]	\label{trivial}

\begin{figure}
	\includegraphics[width=.9\textwidth]{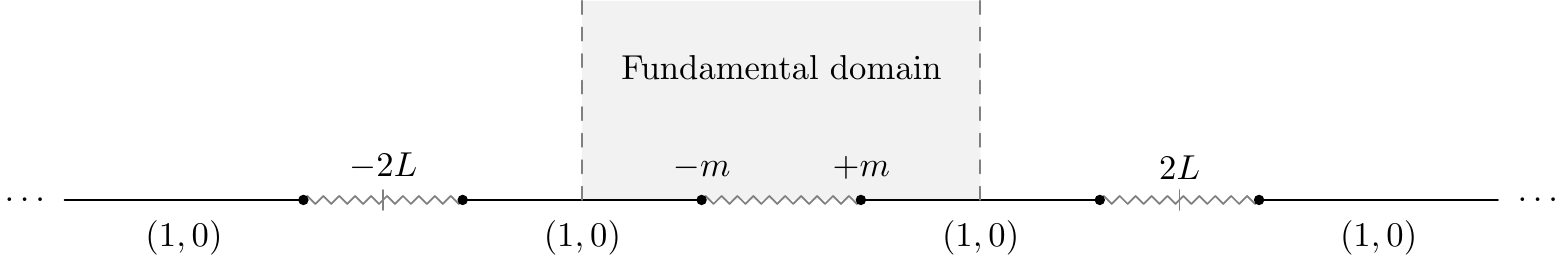}
	\caption{Rod structure for Example \ref{trivial}. Here, as in future rod diagrams, the $z$-axis is drawn horizontally and the zigzag lines indicate horizon rods. The $z$-coordinates of the end points of the horizon rod in the fundamental domain are indicated above the $z$-axis, while the rod structures of the axis rods are indicated below.}
	\label{black-rings}
\end{figure}

We begin with a direct generalization of the Myers-Korotkin-Nicolai construction. It represents an infinite string of identical black rings equally spaced along the $z$-axis. For each $j\in\Z$ consider axis rods $\Gamma_j=\bigl[(2j-2)L+m,2jL-m\bigr]$
with rod structure $(1,0)$ along the $z$-axis, and horizon rods
$I_j=\bigl[2jL-m,2jL+m\bigr]$, see Figure~\ref{black-rings}. Set
\begin{equation}\label{sum}
u = \lim_{n\rightarrow\infty} \left(\sum_{j=-n}^n U_{\Gamma_j} + 2\left(1-\frac{m}{L}\right)\log n\right),\quad\quad\quad
v=0,
\end{equation}
and note that for $(\rho,z)$ fixed and large $j$ we have by \eqref{kilo} that $U_{\Gamma_j}\sim \frac{m-L}{|j|L}$, so that the additional term $2(1-\frac{m}{L})\log n$ renormalizes the divergent series of harmonic functions to yield a finite harmonic limit away from $\Gamma$. Near each $\Gamma_j$ the asymptotics for this function are $u\sim 2\log\rho$, so that the desired rod structure is exhibited by $F$ in \eqref{F11}.
The functions $u$ and $v$ are $2L$-periodic with a fundamental domain $|z|\leq L$.
The fundamental domain contains a single horizon rod $[-m,m]$ where $m<L$, representing horizon topology $S^1\times S^2$. Next observe that since
$u$ and $v$ are even in $z$ with respect to the origin, the right-hand side of the second equation of \eqref{alpha1} is odd in $z$ which yields
\begin{equation}\label{periodic}
\int_{-L}^L \alpha_z\, dz = 0.
\end{equation}
It follows that $\alpha$ is periodic in $z$ with period $2L$. Moreover, by periodicity of the resulting spacetime metric, if the angle defect of a single axis rod vanishes, it will vanish for all. In light of \eqref{b1-b2}, we may then add a constant to $\alpha$ if necessary to achieve this balancing at all axis rods.

In Section \ref{regularity} it will be shown that this solution is asymptotically Kasner. In order to aid with the proof of this statement, here we compute the asymptotics of the harmonic function
\begin{equation}\label{asym0}
u\sim 2\left(1-\frac{m}{L}\right)\log\rho\quad \text{ as }\quad \rho\to\infty.
\end{equation}
To see this, first note that $u$ may be expanded in terms of modified Bessel functions of order zero. It may then be shown that the coefficients of those that have exponential growth must vanish, since basic estimates using the definition \eqref{sum} imply that exponential growth is not possible. From the expansion, we then have that the leading order term is a multiple of $\log\rho$. Consider now the average of $u$ in the $z$-direction
\begin{equation}
\bar{u}=\frac{1}{2L}\int_{-L}^{L} u dz.
\end{equation}
This is a function of the single variable $\rho$ and satisfies
\begin{equation}
\partial_{\rho}^2 \bar{u}+\frac{1}{\rho}\partial_{\rho}\bar{u}=0,
\end{equation}
and so is interpreted as a radial harmonic function in 2-dimensions.
It follows that $\bar{u}=\bar{a}+\bar{b}\log \rho$, for some constants $\bar{a}$, $\bar{b}$. Taking the limit $\rho\to 0$, and using that $u\sim 2\log \rho$ upon approach to the axis rods $\Gamma_j$, we obtain $\rho\partial_{\rho}\bar{u}\rightarrow 2(1-\frac{m}{L})$ which yields the desired conclusion. It then follows from \eqref{alpha1} that
\begin{equation}\label{asym1}
\alpha\sim -\frac{m}{L}\left(1-\frac{m}{L}\right)\log\rho \quad\text{ as }\quad\rho\to\infty.
\end{equation}

Following \cite{khurimatsumotoweinsteinyamada} we may compute the topology of the domain of outer communication. First note that by filling in each horizon rod with an auxiliary axis rod having rod structure $(0,1)$, we obtain an orbit space with an infinite string of alternating $(1,0)$ and $(0,1)$ rods. As will be described in Example \ref{string-corners} below, this is the orbit space for an infinite connected sum $\operatorname \#^{\infty} S^2 \times S^2$. Since each horizon fill-in corresponds to an $S^2 \times B^2$, it follows that the DOC slice topology is given by
\begin{equation}
M^4=
\#^{\infty} S^2\times S^2 \setminus\left[\cup_{i=1}^{\infty} (S^2\times B^2)_i\right].
\end{equation}
Furthermore, we may use periodicity to take quotients and obtain DOCs having the following slice topologies. If the quotient is associated with a single fundamental domain, as in Figure \ref{black-rings}, then upon identifying sides (with the horizon filled in) the orbit space becomes an annulus with outer boundary consisting of two axis rods of rod structure $(1,0)$ and $(0,1)$. This represents $S^4\setminus (T^2 \times B^2)$. If the fundamental domain contains two horizons, then the corresponding slice topology, with horizons filled in, is $S^2\times S^2 \setminus (T^2 \times B^2)$. Increasing the number of horizons in a fundamental domain increases the number of factors in the connected sum of $S^2 \times S^2$ with itself. Finally, removing the horizon fill-ins leads to DOC slice topologies
\begin{equation}
M^4/\!\sim \!\!\!\!\text{ }=M_{k}^4 \setminus\left[(T^2\times B^2) \cup_{i=1}^{k+1}(S^2\times B^2)_i\right],
\end{equation}
where $M_{k}^4$ is $S^4$ if $k =0$, and $\operatorname \#^{k} S^2\times S^2$ if $k \geq 1$.
\end{example}

\begin{remark}
Since $v=0$, it is clear that the space-periodic solution constructed in Example 1 is simply the product of the 4D Myers-Korotkin-Nicolai solutions \cite{myers1987,korotkinnicolai} with a circle of constant length.
\end{remark}

\begin{figure}
	\includegraphics[width=.9\textwidth]{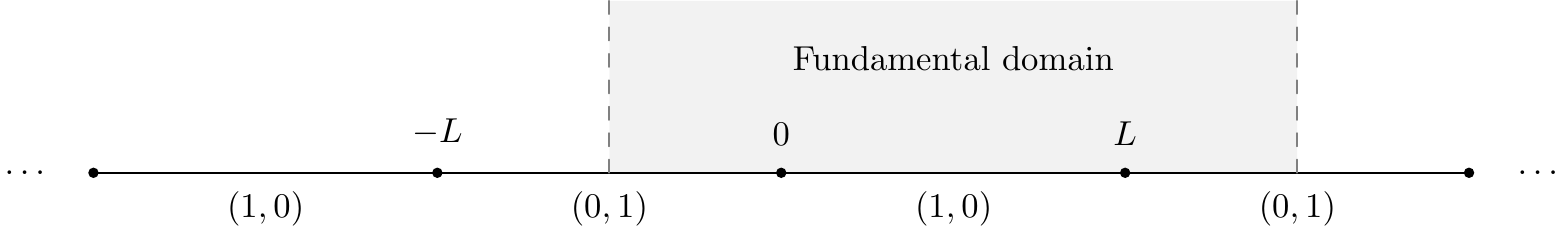}
	\caption{Rod structure for Example \ref{string-corners}.}
	\label{rod-corners}
\end{figure}


\begin{example}[String of corners, the soliton]	\label{string-corners}

To construct the rod structure place a corner at $z=jL$ for each $j\in\Z$ along the $z$-axis. The rod structure will be $(1,0)$ on $\Gamma_{2j}=[2jL,(2j+1)L]$ and $(0,1)$ on $\Gamma_{2j+1}=[(2j+1)L,(2j+2)L]$ for all $j\in\Z$, see Figure \ref{rod-corners}. In order to realize this set of rod structures with the fiber metric $F$ of \eqref{F11}, define harmonic functions
\begin{equation}
u = \lim_{n\rightarrow\infty} \left(\sum_{j=-n}^n U_{\Gamma_{2j}} + \log n\right), \qquad
v = \lim_{n\rightarrow\infty} \left(\sum_{j=-n}^n U_{\Gamma_{2j+1}} +\log n\right).
\end{equation}
As in Example \ref{trivial}, the $\log n$ term renormalizes the sum to ensure convergence. Observe that these functions are periodic in the $z$-direction with period $2L$.
Within the fundamental domain of Figure \ref{rod-corners}, there are two axis rods. It suffices to resolve the conical singularity on these two rods, since this will then be propagated to the remaining rods by periodicity. In order to achieve this, we can arrange for $b_1=b_2=0$ in \eqref{b1-b2} simply by adding appropriate constants to $u$ and $v$. Furthermore, note that the functions $u$ and $v$ are even with respect to the line $z=L/2$, and hence the analogue of \eqref{periodic}
holds here, showing that the function $\alpha$ is $2L$-periodic in the $z$-direction. It follows that we obtain a regular space-periodic solution of the static vacuum Einstein equations without horizons.


Let us now analyze the topology of the DOC time slice $M^4$.
Consider the discrete isometry group $\Z$ acting on $M^4$ by $n: z\mapsto z+2nL$. This action is clearly properly discontinuous, hence the quotient $M^4/\!\sim$ is a Riemannian manifold. The fundamental domain of the covering map $M^4\to M^4/\!\sim$ corresponds in the $\rho z$-half-plane to the horizontal strip $-L/2 \leq z \leq 3L/2$, having with two corners. This is topologically a closed disk with two corners on the outer boundary, with a smaller central disk removed. The closed disk with two corners is the orbit space of $S^4$ under the $T^2$ action. Therefore the topology of the quotient space is $M^4/\!\sim\!\!\!\text{ }=S^4\setminus T^2\times B^2$. This is the case $k=0$ in \eqref{time-slice}.

\begin{figure}
	\includegraphics[width=.6\textwidth]{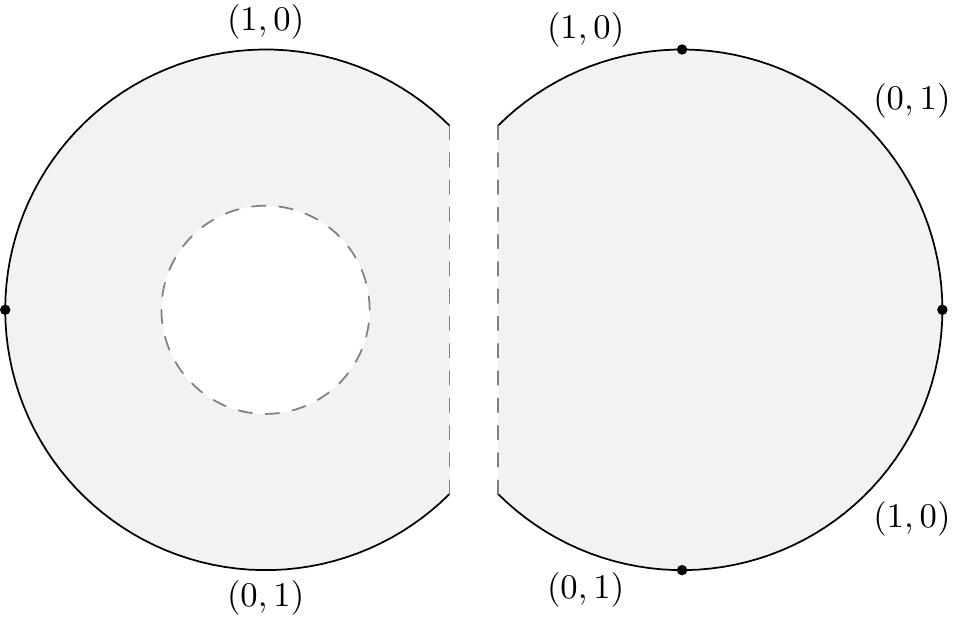}
	\caption{Rod diagram for the connected sum of $S^4\setminus T^2\times B^2$ and $S^2\times S^2$.}
	\label{connect-sum}
\end{figure}

We can also take a quotient by any subgroup $(k+1)\Z$ with $k\geq1$. The fundamental domain will have $2k+2$ corners. Next, it will be shown that the act of adding 2 corners corresponds topologically to taking the connected sum with $S^2\times S^2$. To see this, first note that the rod diagram of $S^2\times S^2$ is a disk with 4 corners, with alternating $(1,0)$, $(0,1)$ rod structures. Moreover, the connected sum operation between two 4D toric manifolds may be described within the orbit space as follows. Namely, excise a corner from each rod diagram so that the cutting curves lift to $S^3$ boundaries and then glue along these cuts, see Figure \ref{connect-sum}. Performing this operation on the two rod diagrams of the figure yields a disk with 4 corners, with a central disk removed. This orbit space arises from the fundamental domain of the group action with $k=1$. Clearly this process can be repeated indefinitely by induction, and corresponds to \eqref{time-slice} with $k\geq1$.

To obtain \eqref{infinite} we need to clarify the definition of an infinite connected sum, since there seems to be some ambiguity on this topic in the literature. Consider an increasing sequence of Riemannian manifolds $N_1\subset N_2\subset \cdots$. Then $N=\cup_j  N_j$, endowed with the natural topology, namely the smallest topology which makes every inclusion $N_j\hookrightarrow N$ continuous, is a Riemannian manifold. This construction may be applied to infinite connected sums as follows. Let $M_i$ be manifolds of the same dimension, and set
\begin{equation}
N_j=\left(\operatorname \#_{i=1}^j M_i\right)\setminus B,
\end{equation}
where $B$ is a ball. Then $N_j$ gives rise to an increasing sequence and we can define $\operatorname \#_{i=1}^\infty M_i = N$. Note that with this definition, the infinite connected sum of $S^2$'s is a plane. To get the punctured plane, we would have to carry out this process twice.

\begin{figure}
	\includegraphics[width=.7\textwidth]{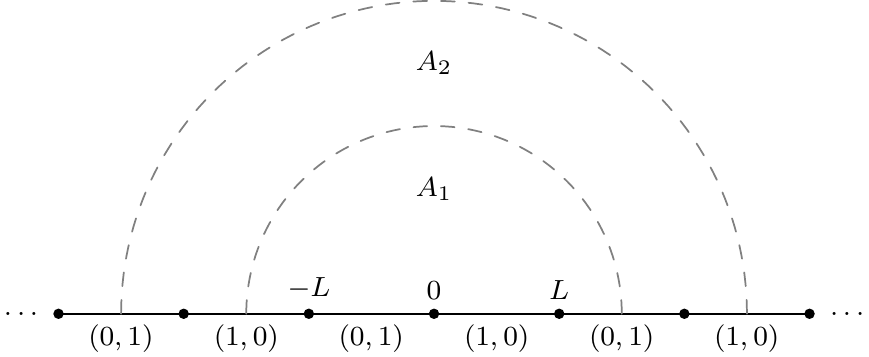}
	\caption{Rod diagram for $\operatorname\#^\infty S^2\times S^2$.}
	\label{inf-connect-sum}
\end{figure}

Now consider the original time slice $M^4$ and its rod diagram, the $\rho z$-half-plane with corners at integer multiples of $L$ on the $z$-axis. Draw concentric half circles centered at the origin with radii $1/2+j$, $j=1,2,\ldots$. These cut the half-plane into infinitely many regions $A_j$, $j=1,2,\ldots$, as in Figure~\ref{inf-connect-sum}. The first one, $A_1$, has three corners and its boundary corresponds to an $S^3$. Therefore it represents $(S^2\times S^2)\setminus B$. Each of the next $A_j$'s, $j\geq2$ has two corners and two $S^3$ boundaries. Hence, as described above $N_j=(\operatorname\#^j S^2\times S^2)\setminus B$. We conclude that the whole half-plane is the quotient space of $\operatorname\#^\infty S^2\times S^2$ under the $T^2$ action.

The space-periodic soliton solution constructed here is asymptotically Kasner. This will be established in the next section. Here, as in Example \ref{trivial}, we simply record the asymptotics of the coefficient functions that appear in the spacetime metric \eqref{spacetimem}. Namely, following the arguments given in the previous example shows that
\begin{equation}\label{asym2}
u\sim\log\rho,\quad\quad\quad v\sim\log\rho,\quad\quad\quad \alpha\sim -\frac{1}{4}\log\rho\quad\quad\text{ as }\quad\quad \rho\rightarrow\infty.
\end{equation}
\end{example}

\begin{example}[String of spheres]
\label{string-spheres}

Consider now the rod structure of Figure \ref{spheres}. All horizon rods (zigzag lines) are of length $2m$ and border two axis rods on both sides, with each axis rod having length $L-m$. The period is then $2L$, and a fundamental domain is shown shaded in grey. The given rod structures imply that the horizon rods represent $S^3$ horizon cross-sections. Thus, this rod diagram described an infinite collection of spherical horizons evenly spaced along the $z$-axis. As in Examples \ref{trivial} and \ref{string-corners}, we can define appropriate $2L$-periodic harmonic functions $u$ and $v$, as renormalized series involving the functions $U_{\Gamma_j}$, in order to realize the desired rod structures. Observe that within the fundamental domain there are precisely two axis rods, one with rod structure $(1,0)$ and the other with rod structure $(0,1)$. Therefore we may use the free constants available within the definitions of $u$ and $v$, in order to resolve the possible conical singularities on these axes by ensuring that $b_1=b_2=0$ in \eqref{b1-b2}.

We now show that $\alpha$, given by \eqref{alpha1}, is $2L$-periodic in the $z$-direction. To do this note that by construction of $u$ and $v$, the following reflection symmetry about the line $z=0$ is manifest
\begin{equation}\label{symmetry}
u(\rho,z) - v(\rho,-z)=c,
\end{equation}
for some constant $c$. It follows that
\begin{equation}
u_\rho(\rho,z) = v_\rho(\rho,-z), \qquad u_z(\rho,z) = -v_z(\rho,-z).
\end{equation}
Thus we have
\begin{equation}
\int_{-L}^0 u_\rho u_z\, dz + \int_0^L v_\rho v_z\, dz = 0, \qquad
\int_{-L}^0 v_\rho v_z\, dz + \int_0^L u_\rho u_z\, dz = 0
\end{equation}
\begin{figure}
	\includegraphics[width=.9\textwidth]{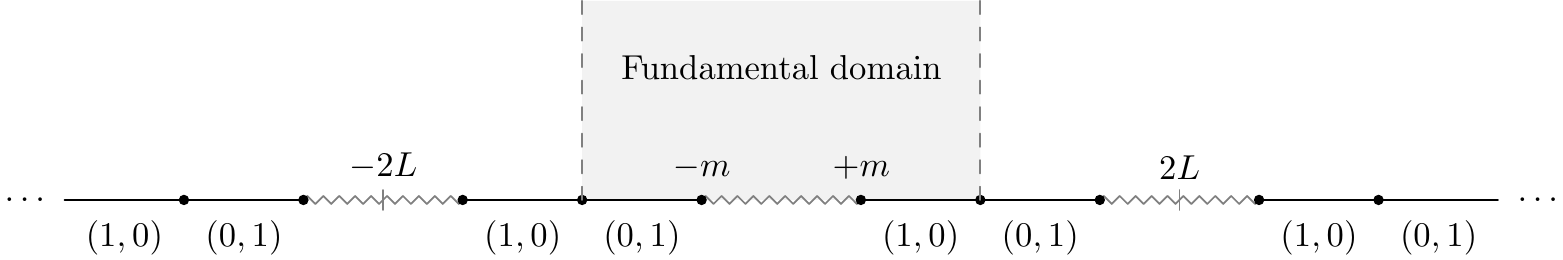}
	\caption{Rod structure for Example~\ref{string-spheres}.}
	\label{spheres}
\end{figure}
Adding these together produces
\begin{equation}
\int_{-L}^L \bigl( u_\rho u_z + v_\rho v_z \bigr) \, dz = 0.
\end{equation}
Moreover, \eqref{symmetry} implies that $u+v$ is even in $z$, and hence $u_\rho+v_\rho$ is even while $u_z + v_z$ is odd. Substituting into \eqref{alpha1} yields
\begin{equation}
\int_{-L}^L \alpha_z\, dz =
\frac\rho4 \int_{-L}^L  \left[ (u_\rho+v_\rho)(u_z+v_z) + \bigl( u_\rho u_z + v_\rho v_z \bigr)
- \frac2\rho (u_z + v_z) \right]\, dz = 0,
\end{equation}
which gives the desired periodicity in the $z$-direction for $\alpha$. The spacetime metric \eqref{spacetimem} arising from $u$, $v$, and $\alpha$ is then space-periodic. Furthermore,
the following asymptotics hold
\begin{equation}\label{jfhg}
u\sim\left(1-\frac{m}{L}\right)\log\rho,\quad v\sim\left(1-\frac{m}{L}\right)\log\rho,\quad \alpha\sim -\frac{1}{4}\left(1+\frac{3m}{L}\right)\left(1-\frac{m}{L}\right)\log\rho
\quad\text{ as }\quad \rho\rightarrow\infty.
\end{equation}
These will be used in the next section to show that the metric is asymptotically Kasner.

To analyze the time slice topology of the domain of outer communication, observe that each horizon rod may be filled-in by extending the two neighboring axis rods and adding a corner in the middle. In the total space this is equivalent to filling-in the horizon with a 4-ball $B^4$. The resulting rod structure is an infinite sequence of alternating $(1,0)$, $(0,1)$ rods, which was shown in the previous example to correspond to an infinite connected sum of $S^2 \times S^2$. Therefore we find that the time slice topology is
\begin{equation}\label{timeslicetop0}
M^4 =\#^{\infty} \> S^2 \times S^2 \setminus \left(\cup_{i=1}^{\infty}B^4_i \right).
\end{equation}
Furthermore, using the periodicity we may take quotients by $(k+1)\mathbb{Z}$
to obtain solutions having slice topology
\begin{equation}\label{7uio}
M^4/\!\sim \!\!\!\!\text{ }=M_{k}^4 \setminus\left[(T^2\times B^2) \cup_{i=1}^{k+1}B_i^4 \right],
\end{equation}
where $M_{k}^4$ is either $S^4$ or $\operatorname \#^{k} S^2\times S^2$ depending on whether $k=0$ or $1\leq k<\infty$.
\end{example}

Our last two examples are presented very briefly since by now the technique is clear. They are very similar to the examples above, and are in fact somewhat simpler since they contain no corners as in Example \ref{trivial}, and do not make use of the more complex symmetry used in Example \ref{string-spheres}.

\begin{example}[String of double spheres] \label{string-double-spheres}

The rod diagram for this example is very similar to the one for Example \ref{trivial}. The only difference is that the rod structure of every other axis rod is changed to $(0,1)$. Consequently the horizons now have topology $S^3$, and the fundamental domain is  $-L\leq z\leq 3L$, see Figure \ref{double-spheres}. There are two axis rods in each fundamental domain, and hence the conical singularities can be cured by choosing the free constants within the definition of $u$ and $v$ appropriately. Furthermore, the functions $u$ and $v$ are clearly
symmetric in the $z$-direction across the line $z=L$. As in previous examples this implies that \begin{equation}
\int_{-L}^{3L}\alpha_{z} dz=0,
\end{equation}
so that $\alpha$ is periodic with period, $4L$, as exhibited by $u$ and $v$. To compute the DOC time slice topology, we can fill-in each horizon rod by extending the neighboring rods and adding a corner in the middle. This leads to \eqref{timeslicetop0}. Moreover, taking a quotient by $(k+1)\mathbb{Z}$ produces \eqref{7uio} with $2(k+1)$ balls removed instead of $k+1$. Lastly, the asymptotics of the relevant functions agree with those of \eqref{jfhg}, which ensures that the solution is asymptotically Kasner.

\begin{figure}
\includegraphics[width=.9\textwidth]{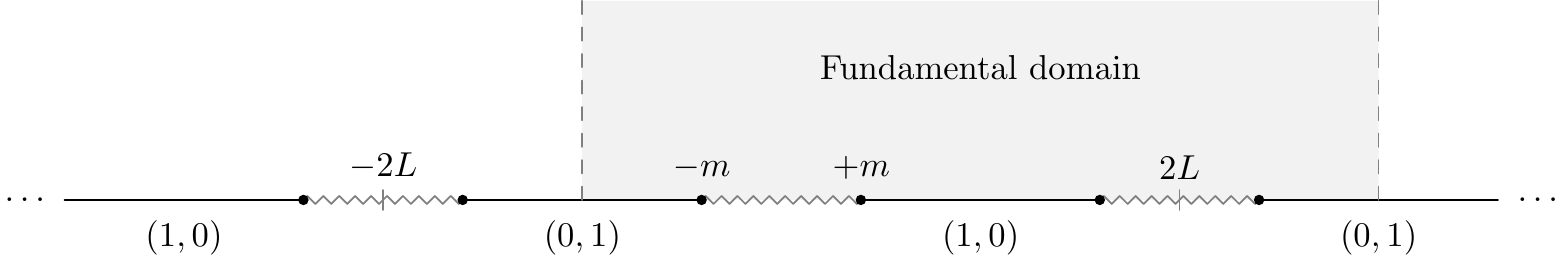}
\caption{Rod structure for Example~\ref{string-double-spheres}.}
\label{double-spheres}
\end{figure}
\end{example}

\begin{example}[String of black Saturns] \label{string-black-saturns}

Consider the rod diagram of Figure \ref{black-saturn}. The fundamental domain is $|z|\leq 3L$,
so that the harmonic functions $u$ and $v$ are $6L$-periodic in the $z$-direction. Furthermore, these functions are clearly symmetric with respect to the line $z=0$, and hence $\alpha$ is periodic with the same period. There are three axis rods in a fundamental domain. One can cure the conical singularities in two of these using the free constants available in the definitions of $u$, $v$, and $\alpha$, and use the symmetry to conclude that the last axis rod is also regular. Horizon rods surrounded by $(1,0)$ and $(0,1)$ axes correspond to $S^3$'s, while horizon rods surrounded by $(0,1)$ axes correspond to $S^1\times S^2$'s. Thus, we obtain an infinite collection of black Saturns, a spherical horizon together with a ring horizon, aligned along the $z$-axis. The spherical horizons may be filled-in with 4-balls and the ring horizons may be filled-in with $B^2\times T^2$'s, resulting in an infinite connected sum of $S^2\times S^2$'s. Hence
\begin{equation}
M^4=
\#^{\infty}\> S^2\times S^2 \setminus\left[\cup_{i=1}^{\infty} B_i^4 \cup_{j=1}^{\infty} (S^2\times B^2)_j\right].
\end{equation}
Moreover, taking a quotient by $(k+1)\mathbb{Z}$ yields
\begin{equation}
M^4/\!\sim \!\!\!\!\text{ }=\#^{k+1} S^2\times S^2 \setminus\left[(T^2\times B^2) \cup_{i=1}^{2(k+1)}B_i^4 \cup_{j=1}^{k+1}(S^2\times B^2)_j\right].
\end{equation}
Lastly the asymptotics are given by
\begin{equation}\label{jfhg11}
u\sim\frac{2}{3}\left(1-\frac{m}{L}\right)\log\rho,\quad v\sim\frac{4}{3}\left(1-\frac{m}{L}\right)\log\rho,\quad \alpha\sim -\frac{1}{9}\left(2+\frac{7m}{L}\right)\left(1-\frac{m}{L}\right)\log\rho
\text{ }\text{ as }\text{ } \rho\rightarrow\infty,
\end{equation}
which guarantees that the solution is asymptotically Kasner.

\end{example}

\begin{remark}
Clearly, many more such examples can be construed. It would be interesting to classify all possible examples in 5D. Furthermore, similar and even more convoluted examples should be possible in higher dimensions. In order to use the same techniques as presented here, for the higher dimensional setting the axisymmetry should be such that the orbit space is still 2-dimensional.
\end{remark}

\begin{figure}
	\includegraphics[width=.9\textwidth]{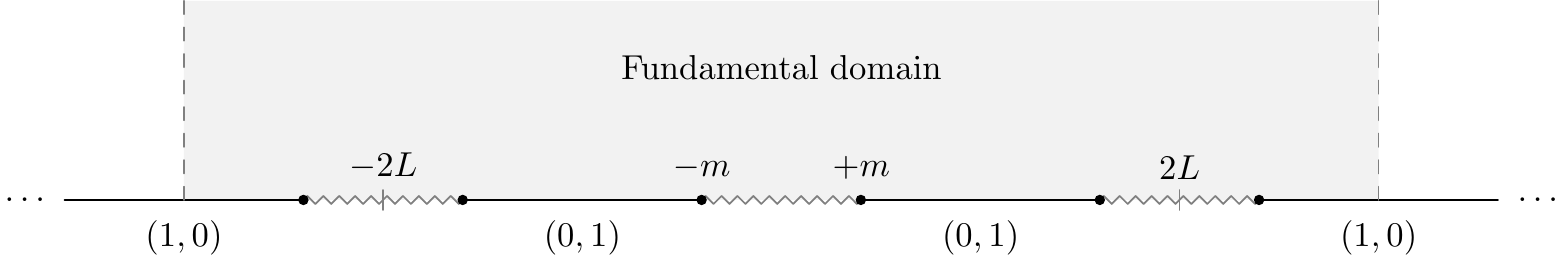}
	\caption{Rod structure for Example~\ref{string-black-saturns}.}
	\label{black-saturn}
\end{figure}

%% file: reg1.tex
\section{Regularity and Asymptotics}
\label{regularity}

In this section, we show that the solutions we construct in Section~\ref{kn} are geometrically regular. This means that the apparent metric singularities that occur on the axes, corners, and poles are simply coordinate singularities, so that after an appropriate change of coordinates the metric is smooth. This will follow from the balancing of conical singularities on the axes. Furthermore, in this section we will show that the solutions constructed exhibit Kasner asymptotics. In addition, the soliton example will be transformed into Riemannian manifolds of nonnegative and zero Ricci curvature, having arbitrarily large second Betti number.

Recall that the relevant spacetime metrics take the form
\begin{equation}\label{metric00}
g=- e^{-u-v} \rho^2 dt^2 +e^{2\alpha}(d\rho^2 +dz^2)+e^u d\phi^2 +e^v d\psi^2,
\end{equation}
on the DOC $\mathbb{R}\times M^4$ where $M^4$ has topology given by Theorems \ref{kng} and \ref{nobody}. Here the two rotational Killing fields are $\partial_{\phi}$ and $\partial_{\psi}$.

\subsection{Regularity}

The metric \eqref{metric00} is clearly static and bi-axisymmetric, and since
$u$ and $v$ are harmonic with $\alpha$ satisfying \eqref{alpha1} this is a solution of the vacuum Einstein equations. Since $u$ and $v$ remain bounded on any horizon rod $\Gamma_h$, we have that $\partial_t$ becomes null there, and hence each such rod represents a Killing horizon. To examine regularity on the interior of horizon rods, observe that integrating \eqref{alpha1} implies $\alpha=-\frac{1}{2}(u+v)-\log c$ on $\Gamma_h$ for a constant $c>0$ that coincides with the horizon surface gravity \cite[Appendix]{hollands2011}. It follows that
\begin{equation}
- e^{-u-v}\rho^2 dt^2 +e^{2\alpha}(d\rho^2 +dz^2)=c^{-2}e^{-u-v}\left(-c^2\rho^2 dt^2 +d\rho^2 +dz^2\right).
\end{equation}
Introducing Kruskal-like coordinates $X,Y>0$ defined by $XY=\rho^2$, $X/Y =e^{2ct}$ then produces
\begin{equation}
g=c^{-2}e^{-u-v}\left(dX dY +dz^2 \right)+e^u d\phi^2 +e^v d\psi^2
\end{equation}
on the interior of $\Gamma_h$. It follows that conical type singularities do not occur on horizon rods and the metric is smooth there, as a consequence of the Einstein equations (and scaling of the time coordinate), rather than the balancing of certain parameters. Regularity of the solution at axis rods, corners, and poles, however, does rely on the balancing of parameters to relieve singularities, as we show below. It should be noted that for general metrics, solid (higher dimensional) cone angles can simultaneously occur independently of whether internal axis conical singularities are present. On the other hand, since the metric \eqref{metric00} satisfies the Einstein equations, in particular due to the relation between $\alpha$ and $u$, $v$, this independence of the two types of cone angles does not occur. That is, balancing of the axis cones implies balancing of the solid cones for the solutions we consider.

\begin{theorem}
The static vacuum spacetimes constructed in Examples \ref{trivial}-\ref{string-black-saturns} are geometrically regular.
More precisely, any degeneracy appearing in the Weyl-Papapetrou expression of the metric \eqref{metric00}, arises solely as a coordinate singularity.
\end{theorem}

\begin{proof}
It remains to check regularity at the interior of axis rods, and in neighborhoods of corners and poles. Consider first a small neighborhood $\mathcal{V}_a \subset\mathcal{M}^5/[\mathbb{R}\times U(1)^2]$ of an interior point to an axis rod $\Gamma_l$, with the property that $\mathcal{V}$ does not intersect the endpoints of $\Gamma_l$. We may assume without loss of generality that the rod structure of $\Gamma_l$ is $(1,0)$. Then $u=2\log\rho+\bar{u}$ within $\mathcal{V}_a\setminus\Gamma_l$, where $\bar{u}$ and $v$ are smooth within $\mathcal{V}_a$. Indeed, this is clear for $v$ since $\mathcal{V}_a$ does not intersect axis rods where $v$ blows-up, and the statement is valid for $\bar{u}$ since $u-2\log\rho$ is bounded and harmonic on $\mathcal{V}_a$ away from a set of codimension 2 \cite[Lemma 8]{weinsteinDMJ}. Substituting $u=2\log\rho+\bar u$ into \eqref{alpha1} produces
\begin{align}
\begin{split}
\alpha_{\rho}=&\frac{1}{2}\bar{u}_{\rho}+\frac{\rho}{4}\left(\bar{u}_{\rho}^2 -\bar{u}_z^2
+v_{\rho}^2-v_z^2+\bar{u}_{\rho}v_{\rho}-\bar{u}_z v_z\right),\\
\alpha_z =&\frac{1}{2}\bar{u}_{z}+\frac{\rho}{4}\left(2\bar{u}_{\rho} \bar{u}_z +2v_{\rho} v_z
\bar{u}_{\rho} v_z +\bar{u}_z v_{\rho}\right).
\end{split}
\end{align}
It follows that $\alpha=\frac{1}{2}\bar{u}+c+O(\rho^2)$ in $\mathcal{V}_a$, for some constant $c$. The spacetime metric within the lift of $\mathcal{V}_a$ may then be expressed as
\begin{equation}
g=-e^{-\bar{u}-v}dt^2 +e^{2c+\bar{u}+O(\rho^2)}(d\rho^2 +dz^2) +\rho^2 e^{\bar{u}} d\phi^2 +e^v d\psi^2.
\end{equation}
Since there is no conical singularity present on $\Gamma_l$, we must have $c=0$. Furthermore, the change of coordinates $x=\rho \cos\phi$, $y=\rho \sin\phi$ yields
\begin{equation}
e^{\bar{u}+O(\rho^2)}d\rho^2 +\rho^2 e^{\bar{u}} d\phi^2=e^{\bar{u}}\left(dx^2 +dy^2\right)
+O(1)\left(xdx+ydy\right)^2.
\end{equation}
It follows that the metric \eqref{metric00} is smooth within the interior of axis rods.

Next consider $\mathcal{V}_c\subset\mathcal{M}^5/[\mathbb{R}\times U(1)^2]$ be a neighborhood of a corner, which separates two axis rods: $\Gamma_1$ to the north (without loss of generality) having rod structure $(1,0)$, and axis rod $\Gamma_2$ to the south having rod structure $(0,1)$. Let $\mathcal{V}_c$ be small enough so that $u=\log(r-z)+\bar u$, and $v=\log(r+z)+\bar{v}$ where $\bar u$ and $\bar{v}$ are smooth. The origin of the coordinate system is centered at the corner point. Insert these expressions into \eqref{alpha1} to find
\begin{align}
\begin{split}
\alpha_{\rho}=&-\frac{\rho}{2r^2}+\frac{(r+z)}{4r}\bar{u}_{\rho}
+\frac{(r-z)}{4r}\bar{v}_{\rho}+\frac{\rho}{4r}(\bar{u}_z -\bar{v}_z)+O(\rho),\\
\alpha_z =&-\frac{z}{2r^2}+\frac{(r+z)}{4r}\bar{u}_{z}
+\frac{(r-z)}{4r}\bar{v}_{z}-\frac{\rho}{4r}(\bar{u}_{\rho} -\bar{v}_{\rho})+O(\rho^2).
\end{split}
\end{align}
On $\mathcal{V}_c$ we then have
\begin{equation}
\alpha=-\frac{1}{2}\log r +\frac{(r+z)}{4r}(\bar{u}-\bar{u}(0))+\frac{(r-z)}{4r}(\bar{v}-\bar{v}(0))+c +O(\rho^2),
\end{equation}
for some constant $c$ and where $\bar{u}(0)$, $\bar{v}(0)$ denote these functions evaluated at the corner. Using these expansions the spacetime metric then becomes
\begin{align}
\begin{split}
g=&-e^{-\bar{u}-\bar{v}}dt^2 +(r-z) e^{\bar{u}} d\phi^2 +(r+z)e^{\bar{v}} d\psi^2\\
&+r^{-1}\exp\left(2c+\frac{(r+z)}{2r}(\bar{u}-\bar{u}(0))
+\frac{(r-z)}{2r}(\bar{v}-\bar{v}(0))+O(\rho^2)\right)(d\rho^2 +dz^2).
\end{split}
\end{align}
The absence of a conical singularities on the two axis rods $\Gamma_1$, $\Gamma_2$ implies that $\bar{u}(0)=\bar{v}(0)$ and $2 e^{2c}=e^{\bar{u}(0)}=e^{\bar{v}(0)}$.
Observe that the geometric angle at the pole between the two axis rods is $\pi/2$, and not $\pi$ as it is in the quotient space. This motivates the change to new coordinates $\xi,\eta\geq 0$ given by $z+i\rho=\frac{1}{2}(\xi+i\eta)^2$, or rather $\rho=\xi\eta$, $z=\frac{1}{2}\left(\xi^2-\eta^2\right)$. We then have
\begin{align}
\begin{split}
g=&-e^{-\bar{u}-\bar{v}} dt^2 +\eta^2 e^{\bar{u}} d\phi^2 +\xi^2 e^{\bar{v}} d\psi^2\\
&+ \exp\left(\bar{u}(0)+\frac{\xi^2}{\xi^2+\eta^2}(\bar{u}-\bar{u}(0))+
\frac{\eta^2}{\xi^2 +\eta^2}(\bar{v}-\bar{v}(0))+O\left(\xi^2 \eta^2\right)\right)(d\xi^2 +d\eta^2).
\end{split}
\end{align}
Next define two pairs of Cartesian coordinates $x_1=\eta \cos\phi$, $y_1=\eta \sin\phi$, and $x_2=\xi \cos\psi$, $y_2=\xi \sin\psi$ with
\begin{equation}
d\eta^2 +\eta^2 d\phi^2=dx_1^2 +dy_1^2, \quad\quad \eta^2 d\eta^2=(x_1 dx_1+y_1dy_1)^2,
\end{equation}
\begin{equation}
d\xi^2 +\xi^2 d\psi^2=dx_2^2 +dy_2^2, \quad\quad \xi^2 d\xi^2=(x_2 dx_2+y_2dy_2)^2.
\end{equation}
It follows that
\begin{align}
\begin{split}
g=& -e^{-\bar{u}-\bar{v}}dt^2 +e^{\bar{u}}\left(\eta^2 d\phi^2 +d\eta^2
+O\left(\frac{\bar{u}-\bar{v}-\bar{u}(0)+\bar{v}(0)}{\xi^2 +\eta^2}+\xi^2\right) \eta^2 d\eta^2\right)\\
&+e^{\bar{v}}\left(\xi^2 d\psi^2 +d\xi^2+ O\left(\frac{\bar{u}-\bar{v}-\bar{u}(0)+\bar{v}(0)}{\xi^2 +\eta^2}+\eta^2\right) \xi^2 d\xi^2\right)\\
=&-e^{-\bar{u}-\bar{v}}dt^2 +e^{\bar{u}}\left(dx_1^2 +dy_1^2 +O\left(\bar{r}^2\right)(x_1 dx_1 +y_1 dy_1)^2 \right)\\
&+e^{\bar{v}}\left(dx_2^2 +dy_2^2 +O\left(\bar{r}^2\right)(x_2 dx_2 +y_2 dy_2)^2 \right),
\end{split}
\end{align}
where $\bar{r}^2=\sum (x_i^2 +y_i^2)$ and we have used the fact that $\bar{u}-\bar{v}$ is smooth at the corner. Therefore the metric \eqref{metric00} is regular near corner points.

Now let $\mathcal{V}_p\subset\mathcal{M}^5/[\mathbb{R}\times U(1)^2]$ be a neighborhood of a pole, separating an axis rod $\Gamma_l$ to the north (without loss of generality) having rod structure $(1,0)$, and a horizon rod $\Gamma_h$ to the south. As before, the neighborhood $\mathcal{V}_p$ may be chosen small enough so that $u=\log(r-z)+\bar u$, where $\bar u$ and $v$ are smooth. Here the origin of the coordinate system is centered at the pole. Substituting this into \eqref{alpha1} produces
\begin{align}
\begin{split}
\alpha_{\rho}=&-\frac{\rho}{2r^2}+\frac{z}{2r}\bar{u}_{\rho}
+\frac{(z-r)}{4r}v_{\rho}+\frac{\rho}{4r}(2\bar{u}_z +v_z)+O(\rho),\\
\alpha_z =&-\frac{z}{2r^2}+\frac{z}{2r}\bar{u}_{z}
+\frac{(z-r)}{4r}v_{z}-\frac{\rho}{4r}(2\bar{u}_{\rho} +v_{\rho})+O(\rho^2).
\end{split}
\end{align}
It follows that in $\mathcal{V}_p$ we have
\begin{equation}
\alpha=-\frac{1}{2}\log r +\frac{z}{2r}(\bar{u}-\bar{u}(0))+\frac{(z-r)}{4r}(v-v(0))+c +O(\rho^2),
\end{equation}
for some constant $c$ and where $\bar{u}(0)$, $v(0)$ denote these functions evaluated at the pole. With these expansions the spacetime metric then takes the form
\begin{align}
\begin{split}
g=&-\frac{\rho^2 e^{-\bar{u}-v}}{r-z}dt^2 +(r-z) e^{\bar{u}} d\phi^2 +e^v d\psi^2\\
&+r^{-1}\exp\left(2c+\frac{z}{r}(\bar{u}-\bar{u}(0))
+\frac{(z-r)}{2r}(v-v(0))+O(\rho^2)\right)(d\rho^2 +dz^2).
\end{split}
\end{align}
The absence of a conical singularity on the axis rod guarantees that $2 e^{2c}=e^{\bar{u}(0)}$.
As above we are motivated to change to new coordinates $\xi,\eta\geq 0$ given by $z+i\rho=(\xi+i\eta)^2$, or rather $\rho=2\xi\eta$, $z=\xi^2-\eta^2$. This yields
\begin{align}
\begin{split}
g=&-2\xi^2 e^{-\bar{u}-v} dt^2 +2\eta^2 e^{\bar{u}} d\phi^2 +e^v d\psi^2\\
&+2 \exp\left(\bar{u}(0)+\frac{\xi^2}{\xi^2+\eta^2}(\bar{u}-\bar{u}(0))-
\frac{\eta^2}{\xi^2 +\eta^2}(\bar{u}+v-\bar{u}(0)-v(0))+O\left(\xi^2 \eta^2\right)\right)(d\xi^2 +d\eta^2).
\end{split}
\end{align}
Next consider Cartesian coordinates $x=\eta \cos\phi$, $y=\eta \sin\phi$, and Kruskal-like coordinates $X,Y>0$ with $XY=\xi^2$, $X/Y=\exp\left(2te^{-\bar{u}(0)-\frac{1}{2}v(0)}\right)$, so that
\begin{equation}
d\eta^2 +\eta^2 d\phi^2=dx^2 +dy^2, \quad\quad \eta^2 d\eta^2=(x dx+ydy)^2,
\end{equation}
\begin{equation}
d\xi^2 -e^{-2\bar{u}(0)-v(0)}\xi^2 dt^2=dXdY,\quad\quad \xi^2 d\xi^2=\frac{1}{4}(YdX+XdY)^2.
\end{equation}
It follows that
\begin{align}
\begin{split}
g=& 2e^{-\bar{u}-v+2\bar{u}(0)+v(0)}\left(dX dY+O\left(\frac{2\bar{u}+v-2\bar{u}(0)-v(0)}{\xi^2 +\eta^2}+\eta^2\right)(YdX+XdY)^2 \right)\\
&+2e^{\bar{u}}\left(dx^2 +dy^2 +O\left(\frac{2\bar{u}+v-2\bar{u}(0)-v(0)}{\xi^2 +\eta^2}+\xi^2\right)(xdx +ydy)^2 \right)
+e^v d\psi^2.
\end{split}
\end{align}
Since $2\bar{u}+v$ is a smooth function at the pole we find that the error terms above are regular, and hence the metric \eqref{metric00} is smooth near horizon poles.
\end{proof}

\subsection{Asymptotics}

Here we show that the solutions constructed in Section~\ref{kn} exhibit Kasner asymptotics. Recall that the Kasner metric on $\mathbb{R}^{n,1}$ is given by
\begin{equation}
g_{K}=-dt^2 +\sum_{i=1}^n t^{2p_i} dx_i^2.
\end{equation}
This satisfies the vacuum Einstein equations precisely when the Kasner conditions hold:
\begin{equation}\label{kasner1}
\sum_{i=1}^n p_i =1,\quad\quad\quad\quad \sum_{i=1}^n p_i^2=1.
\end{equation}
The space-periodic solutions that we produce have metrics of the form \eqref{metric00}, where the asymptotics of the coefficients as $\rho\rightarrow\infty$ are given by
\begin{equation}
u\thicksim A\log\rho,\quad\quad
v\thicksim B\log\rho,\quad\quad\alpha\thicksim C\log\rho,
\end{equation}
with $A,B>0$ and
\begin{equation}\label{cccc}
C=\frac{1}{8}\left(A^2 +B^2 +AB -2(A+B)\right).
\end{equation}
It follows that
\begin{equation}
g\thicksim \rho^{2C}(d\rho^2 +dz^2)-\rho^{2-A-B}dt^2 +\rho^A d\phi^2 +\rho^B d\psi^2.
\end{equation}
Since $g$ is a solution of the vacuum Einstein equations,
it is not surprising that the powers of $\rho$ in the above expression satisfy the Kasner conditions. Indeed, set $\tau=\rho^{C+1}$ and observe that $C+1>0$ together with
\begin{equation}
g\thicksim d\tau^2+\tau^{\frac{2C}{C+1}}dz^2 -\tau^{\frac{2-A-B}{C+1}}dt^2
+\tau^{\frac{A}{C+1}}d\phi^2 +\tau^{\frac{B}{C+1}}d\psi^2.
\end{equation}
It may then be verified that the Kasner conditions \eqref{kasner1} hold for any values of $A$ and $B$, as long as $C$ is given by \eqref{cccc}. In this way, the solutions of Section \ref{kn} are asymptotically Kasner. Note, however, that the role of `time' is played here by the spatial variable $\rho$ when the metric is viewed within the Kasner context.

\medskip

\textbf{Acknowledgements.} This material is based upon work supported by the Swedish Research Council under grant no. 2016-06596 while the authors were in residence at Institut Mittag-Leffler in Djursholm, Sweden during the Fall semester of 2019.